\newif\iflipics
\lipicsfalse 

\iflipics
\documentclass[a4paper,UKenglish,cleveref, autoref, thm-restate,authorcolumns]{lipics-v2019}
\else
\documentclass[11pt]{article}
\fi

\usepackage{amsmath,amsthm,amssymb,xspace,verbatim,multirow,bm,bbm,paralist,enumitem}
\usepackage{array,multirow,multicol,booktabs}

\allowdisplaybreaks[1]

\iflipics
    \usepackage{stmaryrd}
\else
    \usepackage[margin=1in]{geometry}
    \usepackage{txfonts,times}
    \usepackage{paralist}
    \usepackage[T1]{fontenc}
    \usepackage{graphicx}
    \usepackage[usenames,dvipsnames,]{xcolor}
    \usepackage[colorlinks, linkcolor=BrickRed, citecolor=blue]{hyperref}
    \usepackage{thm-restate,cleveref}

    \newtheorem{theorem}{Theorem}[section]
    \newtheorem{lemma}[theorem]{Lemma}
    \newtheorem{corollary}[theorem]{Corollary}
    
    \newtheorem{claim}[theorem]{Claim}
    
    \theoremstyle{definition}  
    
    \theoremstyle{remark}  
\fi

\newtheorem{fact}[theorem]{Fact}
\newtheorem{observation}[theorem]{Observation}

\iflipics
    \makeatletter
    \newcommand{\mypar}[1]{\if@nobreak\else\medskip\fi\noindent{\sffamily\bfseries #1.}~}
    \makeatother
\else
    \newcommand{\mypar}[1]{\medskip\noindent{\bfseries #1.}~}
\fi

\newlist{protocol}{description}{1}
\setlist[protocol]{font=\normalfont\em, itemindent=-2ex, itemsep=1pt}

\newcommand{\hdist}{\widehat{\text{dist}}}

\newcommand{\FF}{\mathbb{F}}
\newcommand{\NN}{\mathbb{N}}
\newcommand{\RR}{\mathbb{R}}
\newcommand{\ZZ}{\mathbb{Z}}

\newcommand{\cB}{\mathcal{B}}

\newcommand{\cH}{\mathcal{H}}

\newcommand{\cV}{\mathcal{V}}

\newcommand{\cX}{\mathcal{X}}
\newcommand{\ta}{\tilde{a}}
\newcommand{\tb}{\tilde{b}}

\newcommand{\tA}{\Tilde{A}}

\newcommand{\tf}{\tilde{f}}

\newcommand{\hg}{\hat{g}}
\newcommand{\hp}{\hat{p}}
\newcommand{\hq}{\hat{q}}
\newcommand{\hT}{\hat{T}}
\newcommand{\hB}{\widehat{B}}

\newcommand{\ba}{\mathbf{a}}
\newcommand{\bb}{\mathbf{b}}
\newcommand{\be}{\mathbf{e}}
\newcommand{\bff}{\mathbf{f}}

\newcommand{\bp}{\mathbf{p}}

\newcommand{\bv}{\mathbf{v}}

\newcommand{\match}{\alpha'}
\newcommand{\fing}{\varphi}

\newcommand{\triCount}{\textsc{TriangleCount}\xspace}
\newcommand{\tricntadj}{\textsc{TriangleCount-Adj}\xspace}
\newcommand{\MaxMatching}{\textsc{MaxMatching}\xspace}
\newcommand{\MIS}{\textsc{MIS}\xspace}
\newcommand{\Acyc}{\textsc{Acyclicity}\xspace}
\newcommand{\Topo}{\textsc{TopoSort}\xspace}

\newcommand{\ShortestPath}{\textsc{st-ShortestPath}\xspace}
\newcommand{\SSSP}{\textsc{SSSP}\xspace}

\newcommand{\CrossEdgeCount}{\textsc{CrossEdgeCount}\xspace}

\newcommand{\InducedEdgeCnt}{\textsc{InducedEdgeCount}\xspace}

\newcommand{\eps}{\varepsilon}
\newcommand{\tO}{\tilde{O}}
\newcommand{\tOmega}{\tilde{\Omega}}

\newcommand{\ang}[1]{\langle{#1}\rangle}

\renewcommand{\ge}{\geqslant}
\renewcommand{\le}{\leqslant}
\renewcommand{\geq}{\geqslant}
\renewcommand{\leq}{\leqslant}
\renewcommand{\b}{\{0,1\}}

\DeclareMathOperator{\tdist}{dist}

\DeclareMathOperator{\out}{out}

\newcommand{\eat}[1]{}

\title{Streaming Verification for Graph Problems: Optimal Tradeoffs and Nonlinear Sketches} 

\iflipics

\titlerunning{Streaming Verification for Graph Problems} 

\author{Amit Chakrabarti}{Dartmouth College, USA}{}{}{Work supported in part by NSF under award CCF-1907738.}

\author{Prantar Ghosh}{Dartmouth College, USA}{}{}{Work supported in part by NSF under award CCF-1907738.}

\author{Justin Thaler}{Georgetown University, USA}{}{}{Work supported by NSF SPX award CCF-1918989, and NSF CAREER award CCF-1845125. Parts of this work were performed while visiting the Simons Institute for the Theory of Computing.}

\authorrunning{A.~Chakrabarti and P.~Ghosh and J.~Thaler}

\Copyright{Amit Chakrabarti and Prantar Ghosh and Justin Thaler}



\category{RANDOM} 

\relatedversion{} 

\supplement{}

\acknowledgements{}

\EventEditors{Jaros{\l}aw Byrka and Raghu Meka}
\EventNoEds{2}
\EventLongTitle{Approximation, Randomization, and Combinatorial Optimization. Algorithms and Techniques (APPROX/RANDOM 2020)}
\EventShortTitle{\mbox{\scriptsize APPROX/RANDOM 2020}}
\EventAcronym{APPROX/RANDOM}
\EventYear{2020}
\EventDate{August 17--19, 2020}
\EventLocation{Virtual Conference}
\EventLogo{}
\SeriesVolume{176}
\ArticleNo{22}

\else 

\author{Amit Chakrabarti\thanks{Department of Computer Science, Dartmouth College. Email: \{ac, prantarg\}@cs.dartmouth.edu. Work supported in part by NSF under award CCF-1907738.} \and Prantar Ghosh$^\fnsymbol{footnote}$ \and Justin Thaler\thanks{Department of Computer Science, Georgetown University. Email: justin.thaler@georgetown.edu. Work supported by NSF SPX award CCF-1918989, and NSF CAREER award CCF-1845125. Parts of this work were performed while visiting the Simons Institute for the Theory of Computing.}}

\fi

\begin{document}

\date{}

\maketitle  

\begin{abstract}
  We study graph computations in an enhanced data streaming setting, where a space-bounded client reading the edge stream of a massive graph may delegate some of its work to a cloud service. We seek algorithms that allow the client to \emph{verify} a purported
  \emph{proof} sent by the cloud service that the work done in the cloud is correct.
  A line of work starting with Chakrabarti et al.~(ICALP~2009) has provided such algorithms, which we call \emph{schemes}, for several statistical and graph-theoretic problems, many of which exhibit a tradeoff between the length of the proof and the space used by the streaming verifier.
  
  This work designs new schemes for a number of basic graph problems---including triangle counting, maximum matching, topological sorting, and single-source shortest paths---where past work had either failed to obtain smooth tradeoffs between these two key complexity measures or only obtained suboptimal tradeoffs. Our key innovation is having the verifier compute certain \emph{nonlinear} sketches of the input stream, leading to either new or improved tradeoffs. In many cases, our schemes in fact provide optimal tradeoffs up to logarithmic factors. 

  Specifically, for most graph problems that we study, it is known that the product of the verifier's space cost $v$ and the proof length $h$ must be at least $\Omega(n^2)$ for $n$-vertex graphs. However, matching upper bounds are only known for a handful of settings of $h$ and $v$ on the curve $h \cdot v=\tilde{\Theta}(n^2)$. For example, for counting triangles and maximum matching, schemes with costs lying on this curve are only known for $(h=\tilde{O}(n^2), v=\tilde{O}(1))$, $(h=\tilde{O}(n), v=\tilde{O}(n))$, and the trivial $(h=\tilde{O}(1), v=\tilde{O}(n^2))$. A major message of this work is that by exploiting nonlinear sketches, a significant ``portion'' of costs on the tradeoff curve $h \cdot v = n^2$ can be achieved.
\end{abstract}

\section{Introduction}

It is far easier to verify a proof than to find one. This intuitively clear
fact has been given precise meanings in several settings\eat{in theoretical
computer science}, leading to such landmark results as the IP $=$
PSPACE~\cite{Shamir92} and PCP Theorems~\cite{AroraLMSS98,AroraS98}. There is
a growing body of work on results of this flavor for space-efficient
computations on large data streams~\cite{Thaler-encyclopedia}. In this setting, a
space-bounded client (henceforth named Verifier) that can only process inputs
in the restrictive {\em data streaming} setting has access to a
computationally powerful entity (henceforth named Prover), such as cloud
computing service, that has no such space limitations. As past work has shown,
many fundamental problems that are intractable in the plain data-streaming
model---in the sense that they cannot be solved using sublinear space---do
admit nontrivial solutions in this Verifier/Prover model, {\em without}
Verifier having to trust Prover blindly.

An algorithm in this model specifies a protocol to be followed by Verifier and Prover so that the former may compute some function $f(\sigma)$ of the input stream $\sigma$. Prover, by performing the specified actions honestly, convinces Verifier to output the correct value $f(\sigma)$. However, if Prover fails to follow the protocol, whether out of malice or error (modeling hardware, software, or network faults in the cloud service), then Verifier is highly likely to detect this and {\em reject}. Past work has considered a few different instances of this setup, such as 
(a)~{\em annotated} data streaming algorithms~\cite{ChakrabartiCMT14}---also called {\em online schemes}---where the parties read $\sigma$ together and the protocol consists of Prover streaming a ``help message'' (a.k.a.~proof) to Verifier either during stream processing and/or at the end; (b)~{\em prescient schemes}~\cite{ChakrabartiCGT14, ChakrabartiCMT14}, which are a variant of the above where Prover knows all of $\sigma$ before Verifier sees it; (c)~{\em streaming interactive proofs} (SIPs)~\cite{ChakrabartiCMTV15,CormodeTY11}, where Verifier and Prover engage in multiple rounds of communication.

This work focuses on the first and arguably best-motivated of these models, namely, online schemes. We simply call them {\em schemes}. We give new and improved schemes for several graph-theoretic problems, including triangle counting, maximum matching, topological sorting, and shortest paths. In all cases, the input is a huge $n$-vertex graph $G$ given as a stream $\sigma$ of edge insertions and/or deletions. While most of our problems have been studied before, we give schemes that (a)~have better complexity parameters, in some cases achieving optimality, and (b)~use cleverer algebraic encodings of the relevant combinatorial problems, often exploiting the ability of a streaming algorithm to compute {\em nonlinear} sketches.

\subsection{Setup, Terminology, and Motivation} \label{sec:setup}
We formalize the setup described above. A scheme for a function $f$ specifies three things: (i)~a space-bounded data streaming algorithm used by Verifier to process the input $\sigma$ and compute a summary $\cV_R(\sigma)$, using random coins $R$; (ii)~a help function used by Prover to send a message $\cH(\sigma)$ to Verifier as a ``proof stream'' after the input stream ends;\footnote{A more general (though seldom used) model allows Prover to send help messages after each data item in $\sigma$.}
and (iii)~an output algorithm $\out_R(\cV_R(\sigma),\cH(\sigma))$ capturing Verifier's work during and after the proof stream, which produces values in range$(f) \cup \{\bot\}$, where an output of $\bot$ indicates ``reject.''
If $\cV_R$ and $\out_R$ run in $O(v)$ bits of space and $\cH$ provides $O(h)$ bits of help, then this scheme is called an {\em $(h,v)$-scheme}. A scheme is interesting if we can use $h > 0$ to achieve a value of $v$ asymptotically smaller than what is feasible or known for a basic streaming algorithm, where $h = 0$.
A scheme is said to have
\begin{itemize}[topsep=4pt,itemsep=0pt]
  \item completeness error $\eps_c$ if $\forall \sigma\, \exists \cH: \Pr_R[\out_R(\cV_R(\sigma),\cH(\sigma)) = f(\sigma)] \ge 1-\eps_c$; 
  \item soundness error $\eps_s$ if $\forall \sigma, \cH': \Pr_R[\out_R(\cV_R(\sigma),\cH'(\sigma)) \notin \{f(\sigma), \bot\}] \le \eps_s$.
\end{itemize}
In designing schemes, we will aim for $\eps_s \le 1/3$, which can be reduced further via parallel repetition in standard ways. We will also achieve perfect completeness, i.e., $\eps_c = 0$.
For an $(h,v)$-scheme we refer to $h$ as its hcost (short for ``help cost'') and $v$ as its vcost (``verification cost''). We use the notation $[h,v]$-scheme as a shorthand for an $(\tO(h),\tO(v))$-scheme.\footnote{%
The notation $\tO(\cdot)$ hides factors polynomial in $\log n$.}

It is intuitive that the parameters $h$ and $v$ are in tension, suggesting that they can be traded off against one another. Most of our algorithms do obtain such tradeoffs. We emphasize that actually obtaining a smooth tradeoff for large ranges of $h$ and $v$ values is not automatic: indeed, an important contribution of this work is to obtain such tradeoffs for problems where past work gave comparable results only for specific settings of $h$ and $v$. 

When studying the results discussed below, it is useful to keep a few cost regimes in mind. We focus on graph problems on $n$-vertex inputs. An $(h,v)$-scheme for such a problem is
{\em sublinear} if $h = o(n^2)$ and $v = o(n^2)$;
{\em frugal} if it is sublinear and achieves the stronger guarantee $v = o(n)$; and
{\em laconic} if it is sublinear and achieves the stronger guarantee $h = o(n)$.

Many graph problems are {\em intractable} in the basic one-pass streaming model, meaning that they provably require $\Omega(n^2)$ space. Past work~\cite{ChakrabartiCMT14} implies that any $(h,v)$-scheme for such a problem must have $hv = \Omega(n^2)$. Thus, an $[h,v]$-scheme with $hv = O(n^2)$ for an intractable problem has achieved an optimal tradeoff, up to logarithmic factors.
All of the problems we consider in this paper (except for counting connected components) 
are intractable for dense graphs (i.e., 
graphs with $\Omega(n^2)$ edges).

Frugal schemes are important when Verifier is so starved for space that it cannot afford to store even a constant fraction of the vertices. They are also very interesting from a theoretical standpoint, since even ``easy'' graph problems require at least $\Omega(n)$ space in the basic streaming model. 
On the other hand, laconic schemes are naturally motivated by settings where Verifier does not have streaming access to the proof and has to store it in full. Consider for example a retail client that uploads transactions to the cloud as they occur. It makes sense to have uploaded even terabytes of information {\em in total} over a long period of time: days, months, or years. However, it might not be reasonable for the cloud to transfer a proof consisting of, say, tens of gigabytes to the client. From a theoretical standpoint, in solving an intractable problem, if Verifier has to store the proof, there is no reason to ever try to reduce vcost to $o(n)$, since hcost will then blow up to $\omega(n)$.

\subsection{Problems, Results, and Comparisons with Related Work} \label{sec:results}

Throughout, the input graph $G$ will be on the fixed vertex set $V = [n] := \{1,\ldots,n\}$ and will have $m$ edges. Many results will be stated in terms of tunable parameters $t,s \in \ZZ^+$ that must satisfy $ts \ge n$. Since bounds are asymptotic, this condition can be read as $ts = n$.

\begin{table*}[!hbt]
\begin{minipage}{\textwidth} 
\centering
\begin{tabular}{c c c c}
\toprule
{\bf Problem}
& {\bf Scheme}  
& {\bf Tradeoff}
& {\bf Reference} \\
\midrule
{}
& {$[h,v]; hv=n^3$} 
& {Suboptimal}
& \cite{ChakrabartiCMT14}
\\
{}
& {$[n^2,1]$} 
& {Optimal} 
& {\cite{ChakrabartiCMT14}}
\\
{}
& {$[n,n]$} 
& {Optimal}
& {\cite{Thaler16}}
\\
{}
& {$[t^3,s^2]$; $ts=n$} 
& {Suboptimal}
& {\cite{ChakrabartiG19}}
\\
{\triCount}
& {$[nt^2,s]$; $ts=n$}
& {}
& {\Cref{firstthm}}
\\
\vspace{3mm}
{}
& {$[t,ns]$; $ts=n$} 
& {Optimal}
& {\Cref{secondthm}}
\\
{}
& {$[mn/\sqrt{v},~v]$} 
& {Suboptimal}
& {\cite{ChakrabartiCGT14}}
\\
\vspace{3mm}
{}
& {$[m+h,v]$; $hv=n^2$} 
& {}
& {\Cref{thm:tri-cnt-sparse}}
\\
{\tricntadj}
& {$[h,v]$; $hv=n^2$} 
& {}
& {\Cref{thm:tri-cnt-adjlist}}
\\

\midrule

{}
&{$[m,1]$} 
&{Optimal}
&{\cite{CormodeMT13}}
\\
{}
& {$[n,n]$} 
& {Optimal}
& {\cite{Thaler16}} 
\\
{\MaxMatching}
& {$[t^3,s^2]$; $ts=n$} 
& {Suboptimal}
& {\cite{ChakrabartiG19}}
\\
{}
& {$[nt,s]$; $ts=n$} 
& {Optimal}
& {\Cref{thm:maxmatch-frugal}}
\\
{}
& {$[\alpha'+h,v]$; $hv=n^2$} 
& {}
& {\Cref{thm:maxmatch-laconic}}
\\
\midrule
{\MIS}
& {$[nt,s]$; $ts=n$} 
& {Optimal}
& {\Cref{thm:mis}}
\\
\midrule
{\Acyc/\Topo}
& {$[m,1]$} 
& {Optimal}
& {\cite{CormodeMT13}}
\\
{}
& {$[nt,s]$; $ts=n$} 
& {Optimal}
& {\Cref{thm:topo}; \Cref{cor:acyc}}
\\
\midrule
{}
& {$[Dnt,s]$; $ts=n$} 
& {}
& {\cite{CormodeMT13}}
\\
{\ShortestPath}
& {$[Kn,n]$}
& {}
& {\cite{ChakrabartiG19}}
\\
{}
& {$[Knt,s]$; $ts=n$} 
& {}
& {\Cref{cor:st-shortestpath}}
\\
\midrule
{Unweighted \SSSP}
\vspace{2mm}
& {$[Dnt,s]$; $ts=n$} 
& {}
& {\Cref{thm:unw-sssp}}
\\
{}
& {$[m+n,1]$} 
& {Optimal}
& {\cite{CormodeMT13}}
\\
{Weighted \SSSP}
& {$[DWn,n]$} 
& {}
& {\Cref{thm:sssp-turnstile}}
\\
{}
& {$[Dn,Wn]$} 
& {}
& {\Cref{thm:sssp-atomic}}
\\
\bottomrule
\end{tabular}
\end{minipage}

\caption{Summary of results on the problems considered in this paper. A scheme is deemed \emph{optimal} if it has help cost at most $h$ and space cost at most $v$ for at least one pair $h, v$ such that $h\cdot v \leq \tilde{O}(L)$, whereas it is known that \emph{any} $(h, v)$ scheme that applies to all graphs requires $h \cdot v \geq \Omega(L)$. A blank space in the \emph{Tradeoff} column indicates that it remains open whether
the scheme can be strictly improved. Here, $\alpha'$ is the size of a maximum matching in the input graph, $K$ is the length of a shortest $v_s$--$v_t$ path, $D$ is the maximum distance from the source to the any other reachable vertex, and $W$ is the maximum weight of an edge.}
\label{table:results}
\end{table*}
\mypar{Triangle Counting}
Our starting point is the triangle counting problem (henceforth, \triCount), studied heavily in past work on graph streaming~\cite{BarYossefKS02, BeraC17, BuriolFLMS06, JhaSP13,JowhariG05, KaneMSS12, McGregorVV16,Thaler16}.
Given a multigraph $G$ as a dynamic stream (i.e., insertions and deletions), the goal is to compute~$T$, the number of triangles in $G$. The exact counting version studied in this paper is an intractable problem in the sense of \Cref{sec:setup}: it requires $\Omega(n^2)$ space in basic streaming.

As noted in \Cref{table:results}, we give several new algorithms for \triCount. Our $[nt^2,s]$-scheme improves upon the best known frugal scheme for the problem: for a fixed hcost $h \ge n$, it improves the vcost from $v^{4/3}$ to $v$, where $v=n^{3/2}/\sqrt{h}$, and for a fixed vcost $v \le n$, it improves the hcost from $n^3/v^{3/2}$ to $n^3/v^2$. Our $[t,ns]$-scheme is not only the first laconic scheme for the problem but also achieves smooth optimal tradeoff in its parameter range; thus, it settles the complexity of the problem in the laconic regime. The $[m+h,v]$-scheme whenever $hv = n^2$ generalizes the $[n^2,1]$-scheme from prior work for any $m$-edge graph (for the setting $h=m$ and $v=n^2/m$) and is interesting in the frugal regime for sparse graphs. 

The problem has also been studied in the adjacency-list model (call it \tricntadj) \cite{BuriolFLMS06,KallaugherMPV19,KolountzakisMPT12, McGregorVV16}, where the stream presents the full neighbor list for each vertex contiguously. We give an $[h,v]$-scheme for any $hv = n^2$ for \tricntadj (again, exact counting). In basic streaming, there is no nontrivial algorithm for computing $T$ exactly, or even approximately when $T$ is small; in fact, under a long-standing conjecture in communication complexity, these problems require $\Omega(m)$ space~\cite{KallaugherMPV19}.

\mypar{Maximum Matching}
There is a recent and ongoing flurry of activity on streaming algorithms for \MaxMatching,
the problem of computing the cardinality $\match(G)$ of a maximum-sized matching%
\footnote{The notation $\match(G)$ is by analogy with $\alpha(G)$, which denotes the cardinality of a maximum independent set of {\em vertices}. It can be found, e.g., in the textbook by West~\cite{WestIGT01}.}
in $G$~\cite{AssadiKL17,ChakrabartiK15,FeigenbaumKMSZ08,GoelKK12,Kapralov13,McGregor05,FarhadiHMRR20,KapralovMNT20}. The exact version of the problem (which is what we study here) is intractable.
For the special case of detecting whether a bipartite graph has a perfect matching, there is a frugal $[nt,s]$-scheme \cite{ChakrabartiCMT14}, which achieves optimal tradeoff. See \Cref{table:results} for previous results for the general problem.

In this work, we give (i)~the first optimal frugal $[nt,s]$-scheme for the general \MaxMatching problem, settling its complexity in the frugal regime, and (ii)~an $[\alpha'+h,v]$-scheme whenever $hv = n^2$, which yields a laconic scheme provided $\alpha'(G) = o(n)$. Obtaining a fully general laconic scheme remains an interesting open problem and we suspect that it will require a breakthrough in exploiting the problem's combinatorial structure.

\mypar{Further Graph Problems and a Common Framework}
We obtain new schemes for the \MIS problem, which asks for an inclusion-wise maximal independent set of vertices; the \Acyc problem, which asks whether the input digraph is acyclic; and the \Topo problem, which asks for a vertex ordering of the input DAG that orients all edges ``forwards.'' In each case, we give an $[nt,s]$-scheme. Recent results show that \MIS \cite{AssadiCK19, CormodeDK19} and \Topo \cite{ChakrabartiGMV20} are intractable in basic streaming, so our schemes are optimal in the frugal regime. Importantly, these schemes, the frugal \MaxMatching scheme, and two of the \triCount schemes all fit a common framework: they boil the problem down to counting the number of edges in one or more induced subgraphs of the input graph. Our scheme for this \InducedEdgeCnt problem could be a useful technical result for future work.

\mypar{Shortest Paths}
The single-source shortest path (SSSP) problem is perhaps the most basic problem in classic graph algorithms. In the streaming setting, even the special case of undirected $v_s$--$v_t$ connectivity in constant-diameter graphs is intractable~\cite{FeigenbaumKMSZ08}. As \Cref{table:results} shows,
our $[Dnt,s]$-scheme for unweighted \SSSP (where $D$ is the maximum distance from the source vertex $v_s$ to any vertex reachable from it) generalizes the result of Cormode et al.~\cite{CormodeMT13} from \ShortestPath to \SSSP. Again, as a corollary, we obtain a $[Knt,s]$-scheme for \ShortestPath, where $K$ is the length of a shortest $v_s$--$v_t$ path. This result generalizes the $[Kn,n]$-scheme of Chakrabarti and Ghosh \cite{ChakrabartiG19} and improves upon the $[Dnt,s]$-scheme of Cormode et al.~\cite{CormodeMT13}, since $K$ can be arbitrarily smaller than $D$. The schemes for the weighted version are interesting for small $D$ and $W$, where $W$ is the maximum weight of any edge.

\subsection{Other Related Works}
Abdullah et al.~\cite{AbdullahDRV16} studied the \triCount and \MaxMatching problems in the stronger SIP model that allows rounds of interaction between Prover and Verifier. For \triCount, they gave a $(\log^2 n, \log^2 n)$-SIP using $\log n$ rounds of interaction. They also designed an $(n^{1/\gamma}\log n,\log n)$-SIP with $\gamma=O(1)$ rounds.     
For the weighted \MaxMatching problem, they gave a $(\rho+n^{1/\gamma'}\log n,\log n)$-SIP using $\gamma$ rounds of interaction, where $\gamma'$ is a linear function of $\gamma$, and $\rho$ is the weight of an optimal matching.

Early works on the concept of annotated streams include Tucker et al.~\cite{Tucker05} and Yi et al.~\cite{YiLCHKS08}, who studied {\em stream punctuations} and {\em stream outsourcing} respectively. Motivated by these works, Chakrabarti et al.~\cite{ChakrabartiCMT14} then formalised the model theoretically as the {\em annotated streaming model} and gave schemes for statistical
 streaming problems including frequency moments and heavy hitters, along with some basic results for graph problems. This non-interactive model was subsequently studied by multiple works including Klauck and Prakash~\cite{KlauckP13}, Cormode et al. \cite{CormodeMT13}, and Chakrabarti et al.~\cite{ChakrabartiCGT14}. Subsequent works considered generalized versions of the model, allowing rounds of interaction. These include {\em Arthur-Merlin streaming protocols} of Gur and Raz \cite{GurR13} and the {\em streaming interactive proofs} (SIP) of Cormode et al.~\cite{CormodeTY11}.  Chakrabarti et al.~\cite{ChakrabartiCMTV15} and Abdullah et al.~\cite{AbdullahDRV16} further studied this generalized setting. We refer to the expository article of Thaler
\cite{Thaler-encyclopedia} for a more detailed survey of this area.


\subsection{Our Techniques}

\mypar{Sum-Check and Polynomial Encodings}
As with much prior work in this area (and probabilistic proof systems more generally),
our schemes are variants of the famous \emph{sum-check protocol} of Lund et al.~\cite{LundFKN92}. 
Specialized to our (non-interactive) schemes, this protocol allows Verifier to make Prover honestly compute $\sum_{x \in \cX} g(x)$ for some low-degree polynomial $g(X)$ derived from the input data and some designated set $\cX$. Verifier has no space to compute $g$ explicitly, nor all values $\ang{g(x): x \in \cX}$, but he can afford to evaluate $g(r)$ at a {\em random} point $r$. The Prover steps in by explicitly providing $\hg(X)$, a polynomial
claimed to equal $g(X)$: this is cheap since $g$ has low degree. Verifier can be convinced of this claim by checking that $\hg(r) = g(r)$.

Hence, the main challenge in applying the sum-check technique is to find a way to encode the data stream problem's output as the sum of the evaluations of a low-degree polynomial $g$ so that Verifier can, in small space, evaluate $g$ at a random point $r$.

\mypar{Sketches: Linearity and Beyond}
A streaming Verifier evaluates $g(r)$ by suitably summarizing the input in a {\em sketch}. Viewing the input as updates to a data vector $\bff = (f_1,\ldots,f_N)$, such a sketch $\bv$ is {\em linear} if $\bv = S\bff$ for some matrix $S \in \FF^{v\times N}$, for some field $\FF$.\footnote{%
This field is finite in the streaming verification literature, whereas traditional data streaming uses $\RR$.} Typically, $S$ is implicit in the sketching algorithm and enables stream processing in $\tilde{O}(v)$ space by translating a stream update $f_i \gets f_i + \Delta$ into the sketch update $\bv \gets \bv + \Delta S \be_i$, where $\be_i$ is the $i$th standard basis vector.
In essentially all prior works on stream verification, the polynomial $g$ was such that
$g(r)$ could be derived from such a linear sketch $\mathbf{v}$.

There is one exception: Thaler~\cite{Thaler16} introduced an optimal $[n,n]$-scheme
for \triCount in which Verifier computes a {\em nonlinear} sketch.\footnote{%
Simliar nonlinearity was used recently in the more powerful model of {\em $2$-pass schemes}~\cite{ChakrabartiG19}.} Roughly speaking, the verifier in Thaler's protocol maintains two $n$-dimensional linear sketches $\bv^{(1)}$ and $\bv^{(2)}$, plus a value $C$ that is \emph{not} a linear function of the input stream but instead depends quadratically on $\bv^{(1)}$ and $\bv^{(2)}$. Moreover, the $j$th increment to $C$ uses information that is available while processing the $j$th stream update, but not after the stream is gone. This is in contrast to linear sketches themselves, where the $j$th sketch update depends only on the $j$th stream update and no others.

\mypar{The Shaping Technique}
Another ubiquitous idea in streaming verification is the {\em shaping technique}, which transforms a data vector into a multidimensional array. This trick realizes $g(X)$ as a summation of an even simpler multivariate polynomial: the latter can be evaluated directly by Verifier at several points, which forms the basis for his sketching. When applied to graph problems, this technique was historically used to reshape the $\binom{n}{2}$-dimensional vector of edge multiplicities. Recently, Chakrabarti and Ghosh~\cite{ChakrabartiG19} introduced the idea of reshaping the graph's {\em vertex space}, rather than just the edge space, thereby transforming the adjacency matrix into a $4$-dimensional array. This trick was crucial to obtaining the first frugal schemes for \triCount and \MaxMatching.

\mypar{Our Contributions}
The new schemes in this work make the following contributions.
\begin{itemize}[topsep=4pt,itemsep=1pt]
  \item We design new polynomial encodings for the graph-theoretic problems we study.
  \item We prominently employ nonlinear sketches, in the above sense, for almost all of our scheme designs.
  \item We use the shaping technique on the vertex space, often combining it with nonlinear sketching, thus expanding the applications of this very recent innovation.
\end{itemize}

Our solutions for \triCount are particularly good illustrations of all of these ideas. Where Thaler's nonlinear-sketch protocol treated each vertex as monolithic, our view of each vertex as an object in $[t]\times [s]$ (for some pair $t, s$ with $t \cdot s=n$) let us do two things. In the laconic regime, we get to use Verifier's increased space allowance in a way that Thaler's protocol cannot, thereby extending his $[n,n]$-scheme to get an optimal tradeoff. In the frugal regime, it is significantly harder to exploit vertex-space shaping because Verifier cannot even afford to devote one entry per vertex in his linear sketches. We overcome this by finding a way for many vertices to ``share'' each entry of each linear sketch (see the string of equations culminating in \cref{eq:tricnt-frugal-inner-poly}), thus extending Thaler's protocol to smoothly trade off communication for space.

We also extend the applicability of nonlinear sketching by identifying many further graph problems for which it yields significant 
improvements. Specifically, in \Cref{sec:edge-count}, we describe two technical problems 
called $\InducedEdgeCnt$and $\CrossEdgeCount$, which are later used as primitives to optimally solve several important graph problems, including $\MaxMatching$. We show how to apply sum-check with a nonlinear Verifier (see, e.g, \cref{eq:nonlinear}) to optimally solve $\InducedEdgeCnt$ and $\CrossEdgeCount$. 

Finally, our schemes for \SSSP feature a different kind of innovation on top of vertex-space shaping and new, clever encodings of shortest-path problems in a manner amenable to sum-check. They overcome the frugal Verifier's space limitation by exploiting the Prover's room to generate a proof stream that mimics an iterative algorithm. For the Verifier to play along with such an iterative algorithm while lacking even one bit of space per vertex, a careful layering of fingerprint-based checks is needed on top of the sum-checks. We hope that our work here opens up possibilities for other instances of porting iterative algorithms to a streaming setting with the help of a prover.


\subsection{Preliminaries} \label{sec:prelims}

In this work, the input graph, multigraph, or digraph is denoted $G$ and defined on a fixed vertex set $V = [n]$. In the {\em vanilla} streaming model, $G$ is given as a stream of $(u,v)$ tokens, where $u,v \in V$: the token is interpreted as an insertion of edge $\{u,v\}$ or directed edge $(u,v)$. If $G$ is edge-weighted, the tokens are of the form $(u,v,w)$, where $w \in \ZZ^+$ is a weight. In the {\em turnstile} streaming model, tokens are of the form $(u,v,\Delta)$, denoting that the quantity $\Delta \in \ZZ$ (which can be negative) is added either to the multiplicity or the weight of the edge $\{u,v\}$.

An important primitive in all our schemes is sketching a data vector by evaluating its low-degree extension at a random point. Let us explain what this means. Suppose our data vector, which has dimensionality~$N$, is shaped into a $k$-dimensional array $f$ with dimensions $(s_1, \ldots, s_k)$, where $s_1 s_2 \cdots s_k \ge N$. Equivalently, we have a function $f$ on domain $[s_1] \times \cdots \times [s_k]$. We work over a suitable finite field\footnote{%
The characteristic of $\FF$ must be large enough to avoid ``wrap around'' problems under arithmetic in $\FF$.}
$\FF$. By Lagrange interpolation, there is a unique polynomial $\tf(X_1, \ldots, X_k) \in \FF[X_1, \ldots, X_k]$ such that
\begin{itemize}[topsep=4pt,itemsep=0pt]
  \item for all $(x_1, \ldots, x_k) \in [s_1] \times \cdots [s_k]$,
  we have $\tf(x_1, \ldots, x_k) = f(x_1, \ldots, x_k)$, and
  \item for all $i \in [k]$, we have $\deg_{X_i} \tf \le s_i - 1$.
\end{itemize}
We call $\tf$ the low-degree {\em $\FF$-extension} of $f$. Since $f \mapsto \tf$ is a linear map, we can write $\tf$ as a linear combination of ``unit impulse'' functions (also known as Lagrange basis polynomials):
\begin{align} \label{eq:unit-impulse}
  \delta_{u_1, \ldots, u_k}(X_1, \ldots, X_k)
    := \prod_{i=1}^k \prod_{x_i \in [s_i] \setminus \{u_i\}} (u_i - x_i)^{-1} (X_i - x_i) \,.
\end{align}
To be precise, $\tf(X_1, \ldots, X_k)
= \sum_{(u_1, \ldots, u_k) \in [s_1] \times \cdots \times [s_k]}
f(u_1, \ldots, u_k)\, \delta_{u_1, \ldots, u_k}(X_1, \ldots, X_k)$.
In particular, if $f$ is built up from a stream of pointwise updates, where the $j$th update adds $\Delta_j$ to entry $(u_1, \ldots, u_k)_j$ of the array, then
\begin{equation} \label{eq:stream-update}
  \tf(X_1, \ldots, X_k)
    = \sum_j \Delta_j\,
    \delta_{(u_1, \ldots, u_k)_j}(X_1, \ldots, X_k) \,.
\end{equation}

\begin{fact} \label{fact:dynamicupdate}
  Given $\bp = (p_1,\ldots,p_k) \in \FF^k$ and a stream of pointwise updates to an initially-zero array with dimensions $(s_1,\ldots,s_k)$, we can maintain the evaluation $\tf(\bp)$ using $O(\log|\FF|)$ space, performing $O(k)$ field arithmetic operations per update. In applications, we usually take $\bp \in_R \FF^k$.\footnote{%
The notation $r \in_R A$ means that $r$ is drawn uniformly at random from the finite set $A$.}
For details and implementation considerations, see Cormode et al.~\cite{CormodeTY11}.
\qed
\end{fact}

Another useful primitive is {\em fingerprinting}, used prominently in our SSSP scheme and subtly in subroutines within other schemes. Its goal is to check equality between two vectors $\ba = (a_1,\ldots,a_N)$ and $\bb = (b_1,\ldots,b_N)$ that are provided via turnstile streams in some possibly intermixed order. This is achieved by checking that $\fing_\ba(r) = \fing_\bb(r)$ for $r \in_R \FF$, where $\fing_\ba(X) = \sum_{j=1}^N a_j X^j$ is the {\em fingerprint polynomial} of $\ba$ and has degree at most $N$.
Both fingerprinting and the eventual uses of \Cref{fact:dynamicupdate} in sum-check protocols depend upon the following basic but powerful result.
\begin{fact}[Schwartz--Zippel Lemma] \label{fact:schwartz-zippel}
\hspace{-1.5mm}  For a nonzero polynomial $P(X_1, \ldots, X_n) \in \FF[X_1, \ldots, X_n]$ of total degree $d$, where $\FF$ is a finite field,
  $\Pr_{(r_1, \ldots, r_n) \in_R \FF^n}
  \left[ P(r_1, \ldots, r_n) = 0\right] \le d/|\FF|$.
\qed
\end{fact}

At various points, we shall use a couple of schemes from Chakrabarti et al.~\cite{ChakrabartiCGT14,ChakrabartiCMT14}.
\begin{fact}[\textsc{subset} and \textsc{intersection} schemes; Prop.~4.1 of \cite{ChakrabartiCMT14} and Thm.~5.3 of \cite{ChakrabartiCGT14}]\label{fact:set-schemes} Given a stream of elements of sets $S,T \subseteq [N]$ interleaved arbitrarily, for any $h,v$ with $hv\geq N$, there are $[h,v]$-schemes to compute $|S\cap T|$ and to determine whether $S \subseteq T$.
\qed
\end{fact}


\section{The Triangle Counting Problem} \label{sec:tri-cnt}

A triangle in a (multi)graph is a set of three edges of the form $\{\{u,v\},\{v,w\},\{u,w\}\}$. The \triCount problem asks for the number of such triangles in the input graph. We solve this problem
for multigraphs given by a turnstile stream, establishing the following two theorems. The first gives
improved (but possibly still not tight)
tradeoffs between hcost~$h$ and vcost $v$ in the parameter regime where $h \geq n$ and $v \leq n$.
The second gives \emph{optimal} tradeoffs (up to logarithmic factors) in the regime where $h \leq n$ and $v \geq n$, based on the known lower bound that $hv$ must be $\Omega(n^2)$.
Both results were previously only known
when $h=\Theta(n)$.

We remind the reader that parameters $t,s \in \ZZ^+$ are tunable, subject to $ts = n$.

\begin{theorem}[Improved frugal schemes] \label{firstthm}
  There is an $[nt^2, s]$-scheme for \triCount.
\end{theorem}

\begin{theorem}[Optimal tradeoff for laconic schemes] \label{secondthm}
  There is a $[t, ns]$-scheme for \triCount. This is optimal up to logarithmic factors.
\end{theorem}

\mypar{Overview of Our Methods}
Consider an adjacency matrix $A$ of a graph on vertex set $V$. The addition of a new edge $\{u,v\}$ creates $\sum_{z \in V} A(u,z) A(v,z)$ new triangles.

Suppose that the input stream consists of $L$ edge updates, the $j$th being $(v_{1j}, v_{2j}, \Delta_j)$; recall that its effect is to add $\Delta_j$ to the multiplicity of edge $\{v_{1j},v_{2j}\}$. Suppose that the cumulative effect of the first $j$ updates is to produce a multigraph $G_j$ whose adjacency matrix is $A_j$ and which has $T_j$ triangles (counting multiplicity). As in Thaler's protocol~\cite{Thaler16}, we can then account for the number of triangles added by the $j$th update:
\begin{align*}
  T_j - T_{j-1} &= \sum_{v_3 \in V} \Delta_j\, A_{j-1}(v_{1j}, v_3)\, A_{j-1}(v_{2j}, v_3) \,.
\end{align*}
As a result, the number of triangles $T$ in the final graph $G = G_L$ is
\begin{align} \label{eq:tricnt-basic-sum}
  T &= \sum_{j \in [L]} \sum_{v_3 \in V} \Delta_j\, A_{j-1}(v_{1j}, v_3)\, A_{j-1}(v_{2j}, v_3) \,.
\end{align}

Our two new families of schemes for \triCount apply the shaping technique to the above equation in two distinct ways, resulting in markedly different complexity behaviors.


\iflipics
\else
\subsection{The Laconic Schemes Regime (Proof of Theorem \ref{secondthm})} \label{sec:tricnt-laconic}
\fi
Let $t, s \in \NN$ be parameters with $ts = n$. We first consider rewriting the variable $v_3$ in \cref{eq:tricnt-basic-sum} as a pair of integers $(x_3, y_3) \in [t] \times [s]$ using some canonical bijection. This shapes each matrix $A_{j-1}$ into a $3$-dimensional array $a_{j-1}$, i.e., a function with domain $[n] \times [t] \times [s]$. Let $\ta$ be the $\FF$-extension of $a$ for a sufficiently large finite field $\FF$ to be chosen later. Then \cref{eq:tricnt-basic-sum} becomes
\begin{align}
  T &= \sum_{j \in [L]} \sum_{x_3 \in [t]} \sum_{y_3 \in [s]}
    \Delta_j\, \ta_{j-1}(v_{1j}, x_3, y_3)\, \ta_{j-1}(v_{2j}, x_3, y_3) 
  = \sum_{x_3 \in [t]} p(x_3) \,, \quad\text{where} \label{eq:tricnt-laconic-fext-sum} \\
  p(X_3) &= \sum_{j \in [L]} \sum_{y_3 \in [s]} \Delta_j\, \ta_{j-1}(v_{1j}, X_3, y_3)\, \ta_{j-1}(v_{2j}, X_3, y_3) \,.
  \label{eq:tricnt-laconic-inner-poly}
\end{align}

By the properties of $\FF$-extensions observed above, we have the bound $\deg p \le 2(t-1)$. We now design our scheme as follows.
\begin{protocol}
  \item[Stream processing.] Verifier starts by picking $r_3 \in_R \FF$. As the stream arrives, he maintains a $2$-dimensional array of values $\ta_{j-1}(v, r_3, y)$, for all $(v,y) \in [n] \times [s]$, using \Cref{fact:dynamicupdate}. He also maintains an accumulator that starts at zero and, after the $j$th update, is incremented by
  $\Delta_j \sum_{y_3 \in [s]} \ta_{j-1}(v_{1j}, r_3, y_3)\, \ta_{j-1}(v_{2j}, r_3, y_3)$.
  By \cref{eq:tricnt-laconic-inner-poly}, the final value of this accumulator is $p(r_3)$. 
  
  \item[Help message.] Prover sends Verifier a polynomial $\hp(X_3)$ of degree $\le 2(t-1)$ that she claims equals $p(X_3)$.

  \item[Verification and output.] Using Prover's message, Verifier computes the check value $C := \hp(r_3)$ and the result value $\hT := \sum_{x_3 \in [t]} \hp(x_3)$. If he finds that $C \ne p(r_3)$, he outputs $\bot$. Otherwise, he believes that $\hp \equiv p$ and accordingly, based on \cref{eq:tricnt-laconic-fext-sum}, outputs $\hT$ as the answer.
\end{protocol}

The analysis of this scheme proceeds along standard lines long established in the literature.
\begin{protocol}
  \item[Error probability.] An honest Prover ($\hp \equiv p$) clearly ensures perfect completeness. The soundness error is the probability that Verifier's check passes despite $\hp \not\equiv p$, i.e., that the random point $r_3 \in \FF$ is a root of the nonzero degree-$(2t-2)$ polynomial $\hp - p$. By the Schwartz--Zippel Lemma (\Cref{fact:schwartz-zippel}), this probability is at most $(2t-2)/|\FF| < 1/n$, by choosing $|\FF|$ large enough.

  \item[Help and Verification costs.] Prover describes $\hat{p}$ by listing its $O(t)$ many coefficients, spending $O(t \log n)$ bits, since each is an element of $\FF$ and $|\FF| = n^{O(1)}$ suffices above. Verifier maintains an $n\times s$ array whose entries are in $\FF$, for a vcost of $O(ns \log n)$.
  Overall, we get a $[t, ns]$-scheme, as required.
\end{protocol}


\iflipics
\else
\subsection{The Frugal Schemes Regime (Proof of Theorem \ref{firstthm})} \label{sec:tricnt-frugal}
\fi
Designing frugal schemes on the basis of \cref{eq:tricnt-basic-sum} is more intricate. This time we rewrite the variables $v_{1j}$ and $v_{2j}$ as pairs $(x_{1j},y_{1j})$ and $(x_{2j},y_{2j})$, each in $[t] \times [s]$ for parameters $t,s$ with $ts = n$. The matrices $A_{j-1}$ are now shaped into $3$-dimensional arrays $b_{j-1}$ that can be seen as functions on the domain $[t] \times [s] \times [n]$. As before, let $\tb$ be an appropriate $\FF$-extension. Working from \cref{eq:tricnt-basic-sum} and cleverly using the ``unit impulse'' function $\delta$ seen in \cref{eq:unit-impulse},
\begin{align}
  T
  &= \sum_{v_3 \in V} \sum_{j \in [L]} \Delta_j\, \tb_{j-1}(x_{1j}, y_{1j}, v_3)\, \tb_{j-1}(x_{2j},
    y_{2j}, v_3) \notag \\
  &= \sum_{v_3 \in V} \sum_{w_1, w_2 \in [t]} \sum_{j \in [L]} \Delta_j\, \tb_{j-1}(w_1, y_{1j}, v_3)\, 
    \tb_{j-1}(w_2, y_{2j}, v_3)\, \delta_{x_{1j}}(w_1)\, \delta_{x_{2j}}(w_2) \notag \\
  &= \sum_{v_3 \in V} \sum_{w_1, w_2 \in [t]} q(w_1,w_2,v_3) \,, \quad \text{where} 
  \label{eq:tricnt-frugal-fext-sum} \\
  q(W_1,W_2,V_3)
  &= \sum_{j \in [L]} \Delta_j\, \tb_{j-1}(W_1, y_{1j}, V_3)\, \tb_{j-1}(W_2, y_{2j}, V_3)\,
    \delta_{x_{1j}}(W_1)\, \delta_{x_{2j}}(W_2) \,. \label{eq:tricnt-frugal-inner-poly}
\end{align}

In contrast to 
\iflipics\else \cref{sec:tricnt-laconic}, \fi 
we have a {\em multivariate} polynomial $q(W_1,W_2,V_3)$. We have the bounds $\deg_{W_1} q \le 2(t-1)$, $\deg_{W_2} q \le 2(t-1)$, and $\deg_{V_3} q \le 2(n-1)$, for a total degree of $O(t+n) = O(n)$. Importantly, the number of monomials in $q$ is at most $(2t-1)^2 (2n-1) = O(nt^2)$. We now present the corresponding scheme and its analysis.

\begin{protocol}
  \item[Stream processing.] Verifier picks $r_1, r_2, r_3 \in_R \FF$. As the stream arrives, he maintains two $1$-dimensional arrays: $\tb_{j-1}(r_1,y,r_3)$ and $\tb_{j-1}(r_2,y,r_3)$, for all $y \in [s]$ (using \Cref{fact:dynamicupdate}). He also maintains an accumulator that starts at zero and, after the $j$th update $(x_{1j}, y_{1j}, x_{2j}, y_{2j})$, is incremented by
  \iflipics
\else
  \[
    \Delta_j\, \tb_{j-1}(r_1, y_{1j}, r_3)\, \tb_{j-1}(r_2, y_{2j}, r_3)\,
    \delta_{x_{1j}}(r_1)\, \delta_{x_{2j}}(r_2) \,.
  \]
  \fi
  By \cref{eq:tricnt-frugal-inner-poly}, the final value of this accumulator is $q(r_1,r_2,r_3)$.
  
  Notice that the accumulator is a nonlinear sketch of the input.
  
  \item[Help message.] Prover sends Verifier a polynomial $\hq(W_1,W_2,V_3)$ that she claims equals $q(W_1,W_2,V_3)$. It should satisfy the degree bounds noted above. He lacks the space to store $\hq$, so she streams the coefficients of $\hq$ in some canonical order. 

  \item[Verification and output.] As $\hq$ is streamed in, Verifier computes the check value $C := \hq(r_1,r_2,r_3)$ and the result value $\hT := \sum_{v_3 \in [n]} \sum_{w_1,w_2 \in [t]} \hq(w_1,w_2,v_3)$. If he finds that $C \ne q(r_1,r_2,r_3)$, he outputs $\bot$. Otherwise, he believes that $\hq \equiv q$ and accordingly, based on \cref{eq:tricnt-frugal-fext-sum}, outputs $\hT$ as the answer.

  \item[Error probability.] As before, we have perfect completeness and by the Schwartz--Zippel Lemma (\Cref{fact:schwartz-zippel}, this time using its full multivariate strength), this soundness error is at most $\deg q/|\FF| = O(n)/|\FF| < 1/n$, by choosing $|\FF|$ large enough.

  \item[Help and Verification costs.] Prover can describe $\hq$ by listing its $O(nt^2)$ coefficients. Verifier maintains two $s$-length arrays.
  Overall, we get an $[nt^2, s]$-scheme, as required.
\end{protocol}


\section{A Technical Result: Counting Edges in Induced Subgraphs} \label{sec:edge-count}

We introduce two somewhat technical, though still natural, graph problems: \InducedEdgeCnt and \CrossEdgeCount. We design schemes for these problems giving optimal tradeoffs (as usual, up to logarithmic factors). These schemes are key subroutines in our schemes for more standard, well-studied graph problems---such as \MaxMatching---considered in \Cref{sec:edge-count-apps}.

The \InducedEdgeCnt problem is defined as follows. The input is a stream of edges of a graph $G=(V,E)$ followed by a stream of vertex subsets $\langle U_1,\ldots U_\ell\rangle$ for some $\ell\in \NN$, where $U_i \subseteq V$ for $i\in [\ell]$. To be precise, the latter portion of the stream consists of the vertices of $U_1$ in arbitrary order, followed by a delimiter, followed by the vertices of $U_2$ in arbitrary order, and so on. The desired output is $\sum_{i=1}^\ell |E(G[U_i])|$, the sum of the numbers of edges in the induced subgraphs $G[U_1],\ldots,G[U_\ell]$. Note that $U_1,\ldots,U_\ell$ need not be pairwise disjoint, so the sum may count some edges more than once.

The \CrossEdgeCount problem is an analog of the above for induced bipartite subgraphs. The input is a stream of edges followed by $\ell$ {\em pairs} of vertex subsets $\langle(U_1,W_1),\ldots,(U_\ell,W_\ell)\rangle$, where $U_i\cap W_i=\varnothing$ for $i\in$~$[\ell]$. The desired output is $\sum_{i=1}^\ell |E(G[U_i,W_i])|$, the sum of the number of cross-edges in the induced bipartite subgraphs $G[U_1,W_1],\ldots,G[U_\ell,W_\ell]$. Note that the $U_i$s (or $W_i$s) need not be disjoint among themselves.

Importantly, in both of these problems, the edges {\em precede} the vertex subsets in the stream. This makes the problems intractable in the basic data streaming model. 
We shall prove the following results. 
\begin{lemma} \label{thm:edge-cnt}
  For any $h,v$ with $hv=n^2$, there is an $[h,v]$-protocol for \InducedEdgeCnt.
\end{lemma}
\begin{lemma} \label{thm:crossedge-cnt}
  For any $h,v$ with $hv=n^2$, there is an $[h,v]$-protocol for \CrossEdgeCount.
\end{lemma}

\mypar{Scheme for \InducedEdgeCnt (Proof of \Cref{thm:edge-cnt})} \label{sec:induced-edge-cnt}
For the given instance, let $M$ denote the desired output and let $A$ be the adjacency matrix of $G$.
For each $i\in \ell$, let $B_i \in \b^V$ be the indicator vector of the set $U_i$, i.e., $B_i(v) = 1 \iff v \in U_i$. Then,
\begin{equation} \label{eq:edgecnt-basic}
  M = \frac12 \sum_{i=1}^{\ell} \sum_{v_1,v_2\in V} B_i(v_1)\, B_i(v_2)\, A(v_1,v_2) \,.
\end{equation}

Let $t,s$ be integer parameters such that $ts = n$. We apply the shaping technique to \cref{eq:edgecnt-basic} by rewriting the variables $v_j$ as pairs of integers $(x_j,y_j) \in [t] \times [s]$, for $j \in \{1,2\}$. This transforms the matrix $A$ into a $4$-dimensional array $a$ and each $B_i$ into a $2$-dimensional array $b_i$. Let $\ta$ and $\tb_i$ be the respective $\FF$-extensions. \Cref{eq:edgecnt-basic} now gives
\iflipics
\else
\begin{align}
  2M 
  &= \sum_{i=1}^{\ell} \sum_{x_1,x_2\in [t]} \sum_{y_1,y_2\in [s]} 
    \tb_i(x_1,y_1)\, \tb_i(x_2,y_2)\, \ta(x_1,y_1,x_2,y_2)
  = \sum_{x_1,x_2\in[t]} p(x_1,x_2) \,,\quad\text{where} \label{eq:edgecnt-poly} \\
  p(X_1,X_2)
  &= \sum_{i=1}^{\ell} \sum_{y_1,y_2\in [s]}
    \tb_i(X_1,y_1)\, \tb_i(X_2,y_2)\, \ta(X_1,y_1,X_2,y_2) \,. \label{eq:edgecnt-sum-poly}
\end{align}
\fi

Our scheme exploits this expression in the same general manner as the analogous expressions for the \triCount schemes from \Cref{sec:tri-cnt} (e.g., Equation \eqref{eq:tricnt-laconic-fext-sum}). Prover sends a bivariate polynomial $\hp(X_1,X_2)$, which is claimed to be $p$, by streaming its coefficients. Since $\deg_{X_j} p \le 2(t-1)$ for $j \in \{1,2\}$, Prover need only send $O(t^2)$ coefficients, for a help cost of $\tO(t^2)$. Verifier computes his output using \cref{eq:edgecnt-poly}, giving perfect completeness.
On the soundness side, Verifier checks the condition $\hp(r_1,r_2) = p(r_1,r_2)$ for randomly chosen $r_1,r_2 \in_R \FF$. By the Schwartz-Zippel Lemma (\Cref{fact:schwartz-zippel}), the probability that he is fooled is at most $\deg p/|\FF| = O(t)/|\FF| < 1/n$, for the right choice of $\FF$.
It remains to describe how exactly Verifier evaluates $p(r_1,r_2)$, which we now address.

\begin{protocol}
  \item[Processing the stream of edges.] This is straightforward: Verifier maintains the $2$-dimensional array of values $\ta(r_1,w,r_2,z)$, for all $w,z\in [s]$, using \Cref{fact:dynamicupdate}.
  \item[Processing the stream of vertex subsets.] Verifier initializes an accumulator to zero and allocates workspace for two arrays of length $s$ with entries in $\FF$. For each $i \in [\ell]$, as the vertices of $U_i$ arrive, he maintains $\tb_i(r_1,z)$ and $\tb_i(r_2,z)$ for each $z \in [s]$, using that workspace. Upon seeing the delimiter marking the end of $U_i$, he computes
  \begin{equation} \label{eq:nonlinear}
    \sum_{y_1,y_2\in [s]} \tb_i(r_1,y_1)\, \tb_i(r_2,y_2)\, \ta(r_1,y_1,r_2,y_2)
  \end{equation}
  and adds this quantity to the accumulator. Note that the workspace is reused when the stream moves on from $U_i$ to $U_{i+1}$. By \cref{eq:edgecnt-sum-poly}, after the last set $U_\ell$ is streamed, the accumulator holds $p(r_1,r_2)$.
  
  \item[Help and verification costs.] We argued above that the hcost is $\tO(t^2)$. Meanwhile, Verifier's storage is dominated by the $s \times s$ array he maintains, leading to a vcost of $\tO(s^2)$.
\end{protocol}

Therefore, we obtain a $[t^2,s^2]$-scheme for any parameters $t,s$ with $ts=n$. In other words, we get an $[h,v]$-scheme for any $h,v$ with $hv=n^2$.

\mypar{Scheme for \CrossEdgeCount (Proof of \Cref{thm:crossedge-cnt})}
Our solution for \InducedEdgeCnt can easily be modified to obtain a protocol for \CrossEdgeCount with the same costs. If $B_i$ and $C_i$ are the indicator vectors of the sets $U_i$ and $W_i$, respectively, then the desired output is
\begin{equation}\label{eq:crossedgecnt-basic}
  M = \sum_{i=1}^{\ell} \sum_{v_1,v_2\in V} B_i(v_1)\, C_i(v_2)\, A(v_1,v_2) \,,
\end{equation}
where we used the fact that each $U_i \cap W_i = \varnothing$. Since \cref{eq:crossedgecnt-basic} has essentially the same form as \cref{eq:edgecnt-basic}, a scheme very similar to the previous one solves \CrossEdgeCount: Verifier simply keeps track of arrays corresponding to $C_i$ alongside ones corresponding to $B_i$.


\section{Maximum Matching and Other Applications of Edge Counting} \label{sec:edge-count-apps}

In this section, we show how  \InducedEdgeCnt and \CrossEdgeCount can be used as subroutines
to solve multiple problems that have been widely studied in the basic and annotated data streaming models. These problems include Maximum Matching, Triangle-Counting, Maximal Independent Set, Acyclicity Testing, Topological Sorting, and Graph Connectivity. For the frugal regime where vcost~$=o(n)$, our schemes are often optimal. We specifically discuss the application to \MaxMatching in \Cref{sec:maxmatch},
\iflipics
\else
and give an account of the other applications in \Cref{sec:otherapps}.
\fi

\subsection{The Maximum Matching Problem} \label{sec:maxmatch}

We give the first optimal frugal scheme for computing the cardinality $\alpha'(G)$ of a maximum matching. As noted in prior works \cite{ChakrabartiG19,Thaler16}, checking whether $\alpha'(G)\geq k$ for some $k$ is not hard, given $\tOmega(k)$ bits of help: Prover can simply send a matching of size $k$ and prove its validity. The interesting part is to verify that $\alpha'(G)\leq k$. For this, as in prior works, we exploit the Tutte--Berge formula \cite{BondyM2008}:
\begin{equation}\label{eq:Tutte-Berge}
    \match(G) = \frac12\, \min_{U\subseteq V} \Big(\, |U| + |V| - \text{odd}(G\setminus U) \,\Big) \,,
\end{equation}
where $\text{odd}(G\setminus U)$ denotes the number of connected components in $G\setminus U$ with an odd number of vertices. Thus, to show that $\alpha'(G)\leq k$, Prover needs to exhibit $U^* \subseteq V$ such that $k=\frac12 (|U^*| + |V| - \text{odd}(G\setminus U^*))$. 
Set $H := G\setminus U^*$. To verify the value of odd$(H)$, the most important sub-check that Verifier must do is to check that all purported connected components of $H$ (sent by Prover) are actually disconnected from each other. Thaler~\cite{Thaler16} gave an $[n,n]$-scheme for this subproblem (thus obtaining the first $[n,n]$-scheme for \MaxMatching), while Chakrabarti and Ghosh~\cite{ChakrabartiG19} gave a $[t^3,s^2]$-scheme for any $ts=n$ (thus designing the first frugal scheme for \MaxMatching, though suboptimal). The latter work notes that all other sub-checks for \MaxMatching can be done by \emph{optimal} frugal schemes (see~\cite{ChakrabartiG19}, Section 4).

\mypar{Optimal Frugal Scheme} To optimally check that the purported connected components of $H$ are indeed disconnected from each other, we use the \InducedEdgeCnt scheme as a subroutine. Prover streams the vertices in $H$ by listing its connected components in some order $\langle U_1,\ldots,U_\ell\rangle$. Verifier uses \Cref{thm:edge-cnt} to count $m_1 := |E(H)|$ (invoking that lemma with a single subset $V(H)$). In parallel, using the same scheme, Verifier computes the sum $m_2 = \sum_{i=1}^\ell |E(G[U_i])|$. The subsets $U_i$ are pairwise disconnected iff $m_2=m_1$, which Verifier checks. The sub-checks of whether $U_i$s are indeed pairwise disjoint (as sets) and whether $U^*\sqcup V(H) = V(G)$ can be done via fingerprinting (as in \cref{sec:prelims}).
\medskip

\noindent
\textit{Help and verification costs.} Prover streams $U^*$ and the vertices in $H$ in a certain order, which adds $O(n\log n)$ bits to the hcost of the \InducedEdgeCnt protocol. The vcost stays the same, asymptotically, giving us an $[n+h,v]$-scheme for \MaxMatching for any $h,v$ with $hv=n^2$. Overall, we have established the following theorem.

\begin{theorem}\label{thm:maxmatch-frugal}
There is an $[nt,s]$-scheme for \MaxMatching. This is optimal up to logarithmic factors,
  since any $(h,v)$-scheme is known to require $hv = \Omega(n^2)$~\cite{ChakrabartiCMT14}.
  \end{theorem}

\mypar{Protocol for Space Larger Than \boldmath$n$} There is no laconic scheme known for the general \MaxMatching problem. The barrier seems to be that a natural witness for the problem is an actual maximum matching of the graph, which can be of size $\Theta(n)$. We show that large maximum matching size $\alpha'(G)$ is indeed the sole barrier to obtaining a laconic scheme. In particular, for any graph $G$, we give a scheme for \MaxMatching with hcost $\alpha'(G)$. This yields a laconic scheme for the case when $\alpha'(G)=o(n)$. 
\iflipics
\else 

\iflipics
\section{Maximum Matching Redux (Proof of Theorem \ref{thm:maxmatch-laconic})} \label{app:max-match-laconic}
In \Cref{sec:maxmatch}, we gave a frugal scheme for \MaxMatching, i.e., one that has hcost $\geq n$ and vcost $\leq n$. Here, we show that we can get a laconic scheme, i.e., one with hcost $=o(n)$ and vcost $=\omega(n)$ as long as the maximum matching size is $o(n)$.

\else
\fi
Let $H = G\setminus U^*$ \iflipics where $U^*$ is as in the \MaxMatching protocol in Section \ref{sec:maxmatch}, \else as above, \fi and let $U_1,\ldots,U_{\ell}$ be the connected components of $H$.
By the Tutte-Berge formula (\cref{eq:Tutte-Berge}), we have $2k = |U^*| + (n - \text{odd}(H))$. This leads to the following observations.

\begin{observation}\label{obs:ustar-size}
$|U^*| =  O(k)$.
\end{observation} 

\begin{observation}\label{obs:spforest-size}
The number of edges in a spanning forest of $H$ is $|V(H)|-\ell \leq n - \text{odd}(H) = O(k)$.
\end{observation}

We now describe our protocol, which is along the lines of the protocol
\iflipics
in Section \ref{sec:maxmatch},
\else
above,
\fi
but this time we crucially use the fact that we are allowing Verifier a space usage of $v \geq n$.

To show that $\alpha'(G) \geq k$, 
Prover sends a matching $M$ of size $k$. Verifier stores $M$ explicitly and checks that it is indeed a matching. Then, he verifies that $M \subseteq E$ using the Subset Scheme (Fact \ref{fact:set-schemes}). Therefore, this part of the scheme uses hcost $\tO(k+h)$ and vcost $\tO(v)$ for any $h,v$ with $hv=n^2$ and $v\geq n$.

Recall 
\iflipics
from \Cref{sec:maxmatch}
\else
\fi
that to show that $\alpha'(G) \leq k$, it suffices to compute odd$(H)$. Prover sends the set $U^*$. By \Cref{obs:ustar-size}, this takes $\tO(k)$ hcost. Verifier has $\Omega(n)$ space, and hence, he can store $V\setminus U^* = V(H)$. Next, Prover sends a spanning forest $F$ of $H$. By \Cref{obs:spforest-size}, this again incurs hcost $\tO(k)$. Verifier stores $F$ and verifies that $F\subseteq E$ using the Subset Scheme (Fact \ref{fact:set-schemes}). From $F$, Verifier explicitly knows the purported connected components $U_1,\ldots, U_\ell$ of $H$.
He finally verifies that $U_i$'s are disconnected from each other by checking that all edges in $H$ are contained in these components. He can do this by checking whether $|E \cap (V(H) \times V(H))| = | E \cap (\cup_{i=1}^\ell U_i \times U_i)|$ using the Intersection Scheme (Fact \ref{fact:set-schemes}). If the check passes he goes over the $U_i$s to compute odd$(H)$ and thus, this part can also be solved using a $[k+h,v]$ scheme for any $h,v$ with $hv=n^2$ and $v\geq n$. Hence, we obtain the following theorem.

\iflipics
\thmlacmm*
\else
\begin{theorem}\label{thm:maxmatch-laconic}
For any $h,v$ with $hv\ge n^2$ and $v\geq n$, there is an $[\alpha'+h,v]$-scheme for \MaxMatching, where $\alpha'$ is the size of the maximum matching of the input graph. In particular, there is an $[\alpha',n^2/\alpha']$-scheme. 
\end{theorem}
\fi

\fi

\iflipics
\else
\iflipics
\section{Applications of Edge-Counting to Other Graph Problems}\label{app:otherapps}
\else\subsection{Applications to Other Graph Problems}\label{sec:otherapps}
\fi

In \Cref{sec:maxmatch}, we used a scheme for \InducedEdgeCnt to obtain an optimal frugal scheme for \MaxMatching. Below, we give applications of edge-counting schemes to several other well-studied graph problems.

\mypar{Triangle-Counting} A scheme for \triCount follows immediately from \InducedEdgeCnt. For $v\in [n]$, set the subsets $U_v = N(v)$, the neighborhood of vertex $v$. Then, observe that \InducedEdgeCnt returns three times the total number of triangles in the graph. The sets $U_v$, however, need to be sent in some order by Prover, and so the additional hcost to \InducedEdgeCnt is $\tO\left(\sum_v |N(v)|\right)=\tO(m)$. As Prover basically repeats the edge stream in a different order, we can check if it's consistent with the input stream by fingerprinting (see \Cref{sec:prelims}). Hence, we get an $[m+h,v]$-scheme for any $h,v$ with $hv=n^2$.

\iflipics
\thmtricntsparse*
\else
\begin{theorem}\label{thm:tri-cnt-sparse}
For any $h,v$ with $hv\ge n^2$, there is an $[m+h,v]$-scheme for \triCount. In particular, there is an $[m,n^2/m]$-scheme.
\end{theorem}
\fi

The only other scheme for \triCount achieving $hv=n^2$ tradeoff with vcost $=o(n)$ was an $[n^2,1]$-scheme by Chakrabarti et al. \cite{ChakrabartiCMT14}. Our result generalizes it for any graph with $m$ edges, thus achieving a better hcost and a smooth tradeoff for sparse graphs.  

We note that in the above scheme, Prover needs to send the sets $U_v=N(v)$ because the \InducedEdgeCnt protocol needs the neighborhood of each vertex to arrive contiguously in the stream. This is essentially the input stream order in the \emph{adjacency-list} or the \emph{vertex-arrival} streaming model. Thus, for the problem \tricntadj, Verifier gets the $U_v$s in the desired order as part of the input; so Prover need not repeat them, saving the huge $\tO(m)$ hcost. However, there is another issue in directly applying the \InducedEdgeCnt subroutine in this case. In the definition of \InducedEdgeCnt, we assume that all the edges in the graph arrive before the vertex subsets $U_i$. Here, the $U_v$s and the edges arrive in interleaved manner (although each $U_v$ arrives contiguously). But we show that we can still apply the scheme for \InducedEdgeCnt to get the desired output. Let the order in which the $U_v$s appear be $\langle U_1, \ldots U_n\rangle$, and let $G_v$ denote the graph consisting of edges seen till the arrival of $U_v=N(v)$. Then, applying \InducedEdgeCnt, what we count is 
\begin{align*}
  \sum_{v\in[n]} |E(G_v[N(v)])|
  = \sum_{v\in[n]} \#\{\text{triangles incident on $v$ in } G_v\}
  = 2T \,.
\end{align*}

The last equality follows since every triangle whose vertices appear in the order $\langle v_1,v_2,v_3\rangle$ will be counted twice: once when $v_2$ arrives and once when $v_3$ arrives.  We therefore obtain the following theorem.

\iflipics
\thmtricntadj*
\else
\begin{theorem}\label{thm:tri-cnt-adjlist}
For any $h,v$ with $hv\ge n^2$, there is an $[h,v]$-scheme for \tricntadj.
\end{theorem}
\fi
%

\mypar{Maximal Independent Set (MIS)} Recent works \cite{AssadiCK19, CormodeDK19} have studied the problem of finding a maximal independent set in the basic data streaming model. They show a lower bound of $\Omega(n^2)$ for a one-pass streaming algorithm. This implies a lower bound of $hv\geq n^2$ for any $[h,v]$-scheme for MIS. Hence, we aim for $hv=n^2$ and describe a frugal scheme using \InducedEdgeCnt. Since the output size of the problem can be $\Theta(n)$, it would only make sense in the frugal regime if the Prover sends the output as a stream and the Verifier checks that it is valid using $o(n)$ space.  

Let $U$ be an MIS in the graph $G$. Prover sends $U$ and Verifier uses \InducedEdgeCnt to count the number of edges in $G[U]$ and verifies that it equals $0$. If the check passes, $U$ is indeed an independent set. It remains to check the maximality of $U$. If $U$ is maximal, then, for each vertex $v$ in $G\setminus U$, there must be a vertex $u$ in $U$, such that $(v,u)$ is an edge. Prover points out such a vertex $u\in U$ for each $v\in G\setminus U$. Let $F$ denote this set of $|G\setminus U|$ purported edges. Now, we use Subset Scheme (Fact \ref{fact:set-schemes}) to verify that $F\subseteq E$, i.e., all these edges are actually present in $G$. We can use fingerprinting (as in \Cref{sec:prelims}) to check that $F$ contains an edge for each vertex in $G\setminus U$ and the Intersection Scheme to verify that the set of their partners is disjoint from $G\setminus U$, i.e., belong to $U$. Thus, the additional hcost to \InducedEdgeCnt, Subset, and Intersection Schemes is $\tO(n)$, the number of bits required to send $U$ and $F$. Therefore, by \Cref{thm:edge-cnt}, we get an $[n+h,v]$-scheme for MIS for any $h,v$ with $hv=n^2$. Thus, our scheme is optimal for the frugal regime.

\iflipics
\thmmis*
\else
\begin{theorem}\label{thm:mis}
For any $t,s$ with $ts=n$, there is an $[nt,s]$-scheme for MIS. This is optimal up to logarithmic factors, since any $(h,v)$-scheme is known to require $hv = \Omega(n^2)$.
\end{theorem}
\fi

\mypar{Acyclicity Testing and Topological Sorting}
We now turn to the \Acyc problem in directed graphs.
It is easy to prove that a graph is \textit{not} acyclic by showing the existence of a cycle~$C$. Verifier checks that $C\subseteq E$ using Subset Scheme (Fact \ref{fact:set-schemes}). Hence, this can be done using an $[h,v]$-scheme for any $h\geq |C|$. 

The more interesting case is when the graph is indeed acyclic. Note that a directed graph is acyclic if and only if it has a topological ordering. Thus, it suffices to show a valid topological ordering of the vertices. \Topo is a fundamental graph algorithmic problem of independent interest. \Acyc has a one-pass lower bound of $\Omega(n^2)$ in the basic data streaming model. Recently, Chakrabarti et al. \cite{ChakrabartiGMV20} showed that \Topo also requires $\Omega(n^2)$ space in one pass. These translate to a lower bound of $hv\geq n^2$ for any $[h,v]$-scheme for these problems. Hence, we aim for a scheme with $hv=n^2$ and design a protocol for \Topo in the frugal regime. Since this problem has output size $\tilde{\Theta}(n)$, we aim for a protocol where Prover sends a topological ordering of the graph and Verifier checks its validity using $o(n)$ space. Moreover, this protocol can be used for the YES case of \Acyc. 

Verifier uses \CrossEdgeCount to solve this. As Prover sends the topological order $\langle v_1,\ldots,v_n\rangle$, for each $i\in[n-1]$, Verifier sets $U_i = \{v_1,\ldots,v_i\}$ and $W_i=\{v_{i+1}\}$ for \CrossEdgeCount. Thus, the protocol counts precisely the number of forward edges induced by the ordering. If it equals $m$, then the ordering is indeed a valid topological order. Note that since $U_{i+1}=U_i\cup \{v_{i+1}\}$, Prover doesn't need to send $U_{i+1}$ afresh; just $v_{i+1}$ is enough for Verifier to update his sketch. Verifier can use fingerprinting (see \Cref{sec:prelims}) to make sure that precisely the set $V$ was sent in some order. Hence, the additional hcost to \CrossEdgeCount is the number of bits required to express the topological order, i.e., $\tO(n)$. Therefore, by \Cref{thm:crossedge-cnt}, we get a $[n+h,v]$-scheme for any $hv=n^2$.  

\iflipics
\thmtopo*
\thmacyc*
\else
\begin{theorem}\label{thm:topo}
For any $t,s$ with $ts=n$, there is an $[nt,s]$-scheme for \Topo. This is optimal up to logarithmic factors,
  since any $(h,v)$-scheme is known to require $hv = \Omega(n^2)$.
\end{theorem}
\begin{corollary}\label{cor:acyc}
For any $t,s$ with $ts=n$, there is an $[nt,s]$-scheme for \Acyc. This is optimal up to logarithmic factors,
  since any $(h,v)$-scheme is known to require $hv = \Omega(n^2)$.
\end{corollary}
\fi
For dense graphs, our result generalizes the $[m,1]$-scheme of Cormode et al. \cite{CormodeMT13} for \Acyc by achieving a smooth tradeoff.  

\mypar{Graph Connectivity} The graph connectivity problem has garnered considerable attention in the basic and annotated streaming settings \cite{AhnGM12, ChakrabartiCMT14, Thaler16}. For any $t,s$ with $ts=n$, Chakrabarti et al. \cite{ChakrabartiCMT14} gave an $[nt,s]$-scheme  that determines whether an input graph is connected or not. Their scheme cannot, however, solve the more general problem of returning the number of connected components. The $[t^3,s^2]$-scheme (for any $ts=n$) of Chakrabarti and Ghosh \cite{ChakrabartiG19} does solve this problem, but has a worse tradeoff. As noted in \Cref{sec:maxmatch}, we can use \InducedEdgeCnt to check that all purported connected components are indeed disconnected from each other. On the other hand, the scheme of Chakrabarti et al. \cite{ChakrabartiCMT14} can check whether each component is actually connected. Hence, we can verify the number of connected components claimed by Prover by running these schemes parallelly. Thus, we generalize the result of Chakrabarti et al. \cite{ChakrabartiCMT14} by obtaining an $[nt,s]$-scheme for counting the number of connected components of a graph.  

\iflipics
\thmconn*
\else
\begin{theorem}\label{thm:connectivity}
For any $t,s$ with $ts=n$, there is an $[nt,s]$-scheme for counting the number of connected components of a graph.
\end{theorem}
\fi
\fi

\section{The Single-Source Shortest Path Problem}
\label{sec:sssp}

In the single-source shortest path (\SSSP) problem, the goal is to find the distances from a source vertex $v_s$ to every other vertex reachable from it. In \Cref{sec:unwt-sssp}, we give a $[Dnt,s]$-scheme for the unweighted version, whenever $ts=n$. If $s=o(n)$, Verifier does not have enough space to store the output; therefore, we aim for a protocol where Prover streams the output, and Verifier checks that it is correct using $o(n)$ space, thus achieving a frugal scheme. 

In \Cref{sec:sssp-weighted}, we state our results for weighted \SSSP for the two different weight update models
described in \Cref{sec:prelims}
\iflipics
\else
: (i)~a $[DWn,n]$-scheme for the ``turnstile'' model, and (ii)~a $[Dn,Wn]$-scheme for the ``vanilla'' model. 
\fi
\subsection{Unweighted \SSSP}\label{sec:unwt-sssp}

We shall design a scheme that works even if the same edge appears multiple times in the stream (unlike prior work~\cite{CormodeMT13} that assumes that an edge appears at most once). 

Prover sends
distance labels $\hdist[v]$ for all $v \in V$, claiming that $\hdist[v] = \tdist(v_s,v)$, 
the actual distance from the source vertex $v_s$ to $v$. 
Let the radius-$d$ ball around $v_s$ be $B_d := \{v \in V:\, \tdist(v_s,v)\leq d\}$ and let
$\cB:=\{B_d:\, d \in [D]\}$ be the family of such balls. Let $\hB_d$ be the corresponding balls implied by Prover's $\hdist$ labels, and $\widehat{\cB}:=\{\hB_d: d \in [D]\}$.

To check correctness, Verifier uses fingerprinting (\Cref{sec:prelims}) modified as follows. Letting $B,\hB$ also denote the respective characteristic vectors, define fingerprint polynomials
\[
  \fing_{\cB}(X,Y):=\sum_{i\in [n]}\sum_{d\in [D]} B_d(i) X^i Y^d \,,
  \quad
  \fing_{\widehat{\cB}}(X,Y):=\sum_{i\in [n]}\sum_{d\in [D]} \hB_d(i) X^i Y^d \,,
\]
As the 
$\hdist$ labels are streamed, Verifier 
constructs the fingerprint $\fing_{\widehat{\cB}}(\beta_1, \beta_2)$ for some $\beta_1, \beta_2 \in_R \FF$. 

Over the course of the protocol, using further help from Prover, Verifier will construct the sets $B_d$ inductively and, in turn, the ``actual'' fingerprint $\fing_{\cB}(\beta_1,\beta_2)$. The next claim shows that comparing this with $\fing_{\widehat{\cB}}(\beta_1,\beta_2)$ validates Prover's $\hdist$ labels.

\begin{claim}\label{clm:SSSP correctness}
If $\hB_d = B_d$ for all $d$, then $\hdist[v] = \tdist(v_s,v)$ for all vertices $v$.
\end{claim}
\begin{proof}
Suppose not. Let $d^*$ be the smallest $d$ such that $\exists\, u\in B_{d^*}$ with $\hdist[u]\neq \tdist(v_s,u)$. Therefore, $\tdist(v_s,u)=d^*$. Now, $d^*$ cannot be $0$ since $v_s$ is the only vertex in $B_0$ and Verifier would reject immediately if $\hdist(v_s) \neq 0$. 
Since $B_{d^*} = \hB_{d^*}$, we have $u\in \hB_{d^*}$. This means $\hdist(u)\leq d^*$. Since $\hdist(u)\neq d^*$, we have $\hdist(u)\leq d^*-1$. Thus, $u\in \hB_{d^*-1}$, i.e., $u\in B_{d^*-1}$, which is a contradiction to the minimality of $d^*$.  
\end{proof}

As before, $A$ denotes the adjacency matrix of the graph. Putting
\begin{gather}
  q_d(u) := \sum_{v\in V} B_d(v)\, A(v,u) \,, \text{ for each } u \in V \,, \label{eq:q-basic} \\
  \text{we have~~} B_{d+1} = \left\{u \in V:  q_{d}(u) \neq 0\right\}\,. \label{eq:unw-nbhd}
\end{gather}
To apply the shaping technique to \eqref{eq:q-basic}, rewrite $v$ as $(x,y) \in [t] \times [s]$. This reshapes $A$ into a
$t\times s\times n$ array $a(x,y,u)$ and $B_d$ into a $t\times s$ array $b_d(x,y)$. As usual, let $\ta$ and $\tb_d$ be the respective $\FF$-extensions for a suitable finite field $\FF$. Then, \cref{eq:q-basic} gives
\begin{align}
    q_d(u) &= \sum_{x\in [t]} p_{d}(x,u) \,, \quad\text{where} \label{eq:q-sum-poly} \\
    p_d(X,U) &:= \sum_{y \in [s]} \tb_d(X,y)\, \ta(X,y,U) \,. \label{eq:q-tradeoff-poly}
\end{align}

\begin{protocol}
\item[Stream processing.] Verifier picks $r_1,r_2\in_R \FF$ and maintains $\ta(r_1,y,r_2)$. When he sees vertices in $B_1$, i.e., $v_s$ and its neighbors, he maintains $b_1(r_1,y)$ for all $y\in [s]$ and also updates the fingerprint $\fing_{\cB}(\beta_1,\beta_2)$ accordingly.

\hspace{3ex} Verifier wants to construct the values $b_d(r_1,y)$ inductively for 
$d\in [D]$. For constructing $b_{d+1}$ values for some $d$, he wants all $u$ such that $q_d(u) \neq 0$ (\cref{eq:unw-nbhd}) in streaming order since he doesn't have enough space to either store the entire polynomial of degree $n-1$ that agrees with $q_d$  (so as to go over all evaluations), or to parallelly evaluate it at $n$ values while its coefficients are streamed. Hence, he asks for the following help message.

\item[Help message processing.] Prover continues her proof stream by sending $\langle \hp_1, Q_1,\ldots, \hp_{D}, Q_{D}\rangle$, where $Q_d := \langle \hq_d(u) : u\in V\rangle$, claiming that $\hp_d \equiv p_d$ and $\hq_d(u) = q_d(u)$ for each $d\in [D]$ and $u\in [n]$.

While $\hp_d$ is streamed, Verifier computes the following in parallel:
\begin{itemize}
    \item $\hp_d(r_1,r_2)$;
    \item $p_d(r_1,r_2)$, using \cref{eq:q-tradeoff-poly};
    \item the fingerprint $g_d := \sum_{u\in [n]}\sum_{x\in [t]} \hp_d(x,u)\beta^u$ (for some $\beta \in_R \FF$).
\end{itemize}
After reading $\hp_d$, he checks whether $\hp_d(r_1,r_2)=p_d(r_1,r_2)$. If so, he believes that $\hp_d \equiv p_d$ and, in turn, that $g_d = \sum_{u\in [n]} q_d(u)\beta^u$ (by \cref{eq:q-sum-poly}). Next, as $Q_d$ is streamed, 
\begin{itemize}
    \item Verifier computes the fingerprint $g_d' := \sum_{u\in [n]} \hq_d(u)\beta^u$.
    \item For each $u$ with $\hq_d(u)\neq 0$, due to \cref{eq:unw-nbhd} (and assuming for now that the $\hq_d$ values are correct), he treats $u$ as a stream update for $B_{d+1}$, and (i) maintains $b_{d+1}(r_1,y)$ for all $y\in [s]$, and (ii)  accordingly updates the fingerprint  $\fing_{\cB}(\beta_1,\beta_2)$.
\end{itemize}
    After reading $Q_d$, he checks if the fingerprints $g_d$ and $g_d'$ match. If they do, he believes that all $\hq_d$ values in $Q_d$ were correct and hence, the $b_{d+1}$ values he constructed are correct as well. He moves on to the next iteration, i.e., starts reading $\hp_{d+1}$.
     
\item[Final Verification.]   After the $D$th iteration, Verifier checks if the two fingerprints $\fing_{\cB}(\beta_1,\beta_2)$ and $\fing_{\widehat{\cB}}(\beta_1,\beta_2)$ match. If the check passes, then he believes that the $\hdist$ labels were correct, at least upto distance $D$ (by \Cref{clm:SSSP correctness}). Finally, he checks if fingerprints for $B_D$ and $B_{D+1}$ match to verify that vertices in $V \setminus B_D$ are indeed unreachable.

\item[Error probability.] Verifier does $O(D)$ fingerprint-checks and $O(D)$ sum-checks, using degree-$O(n)$ polynomials. Using $|\FF|>n^3$ (and a union bound), the soundness error is $< 1/n$.
   
\item[Help and verification costs.] The set of $\hdist$ labels sent by the Prover has size $\tO(n)$. Each polynomial $\hp_d$ has $nt$ monomials and each $Q_d$ has $O(n)$ field elements, and hence, size $\tO(n)$. Therefore, the total hcost is $\tO(Dnt)$. Initially, the $\tA$ and $\tb_1$ values are stored using $\tO(s)$ space. Next, the $\tb_d$ and $g_d$ values are maintained reusing space of $b_{d-1}$ and $g_{d-1}$ values respectively. We also use $O(1)$ many other fingerprints that take $O(\log n)$ space each. Hence, the total vcost is $\tO(s)$.
\end{protocol}

\begin{theorem}\label{thm:unw-sssp}
There is a $[Dnt,s]$-scheme for unweighted \SSSP, where $D = \max\limits_{v \in V} \tdist(v_s,v)$.
\end{theorem}

\begin{corollary}\label{cor:st-shortestpath}
There is a $[Knt,s]$-scheme for \ShortestPath, where $K = \tdist(v_s,v_t)$.
\end{corollary}

\begin{proof}
The protocol for \SSSP incurs a factor of $D$ in the hcost since it constructs $B_d$ for each $d\in [D]$. For the simpler \ShortestPath problem, we can inductively construct balls and stop as soon as we find the destination vertex $v_t$ in some $B_d$ (i.e., get $\hq_{d-1}(v_t)\neq 0$). We must find it in $B_K$ where $K$ is the length of a shortest $v_s$--$v_t$ path. Thus, we will only incur a factor of $K$ in the hcost, which implies a $[Knt,s]$-scheme for \ShortestPath.
\end{proof}

Thus, we generalize the $[Dnt,s]$-scheme of Cormode et al.~\cite{CormodeMT13} from \ShortestPath to \SSSP. Our result for \ShortestPath generalizes the $[Kn,n]$-scheme of Chakrabarti and Ghosh \cite{ChakrabartiG19} by giving a smooth tradeoff and also improves upon the $[Dnt,s]$-scheme of Cormode et al.~\cite{CormodeMT13}, since $K$ can be arbitrarily smaller than $D$.

\subsection{Weighted SSSP}\label{sec:sssp-weighted}

\iflipics
\else
\iflipics
\section{Details of Weighted SSSP Schemes} \label{app:wtd-sssp}

Here, we give full details of the schemes for weighted SSSP in our two streaming models, thereby proving \Cref{thm:sssp-turnstile,thm:sssp-atomic}.
\else
Here, we consider the general weighted version of \SSSP and give schemes for the problem in the vanilla streaming model as well as the turnstile weight update model.  
\fi

\mypar{Turnstile weight update\iflipics (Proof of \Cref{thm:sssp-turnstile})\else\fi}
Assume that the
edge weights are positive integers. Each stream update increments/decrements the weight of an edge. The
\emph{distance} from vertex $u$ to vertex~$v$ refers to the weight of the shortest
path from $u$ to $v$. Let $D$ be the
longest distance from the source $s$ to any
other vertex reachable from it, and $W$ be
the maximum weight of an edge. 

Define
\[
  \delta_w(X) := 
  {\prod\limits_{\substack{w'\in [W]\\ w'\neq w}} (X-w')} \bigg/ {\prod\limits_{\substack{w'\in [W]\\ w'\neq w}} (w-w')} \,.
\]

Let $A$ denote the adjacency matrix of the
weighted graph $G$, i.e., $A(u,v)$ is
the weight of the edge $(u,v)$. Let $B_d$ (resp.~$N_d$) denote the set of vertices at a distance of at most (resp.~exactly) $d$ from the source vertex $v_s$. Then,
\begin{gather}
    N_{d+1} = \left\{u \in V\setminus B_d:\, p_{d}(u) \neq 0\right\} \,, \label{eq:turn-exact-nbhd} \\
    \text{where~~}
    p_d(U) = \sum_{v\in B_d} \delta_{w(v)}(\tA(v,U))
    \text{~~and~~} w(v)=d+1-\tdist[v] \,.
    \label{eq:dist-poly}
\end{gather}

\begin{protocol}
\item[Stream processing.] Verifier chooses
$r\in_R \FF$ and maintains $\tA(v,r)$ for
all $v$. He stores $B_1$ with $\text{dist}[v]$ labelled as $1$ for each $v\in B_1$.
\item[Help message processing and verification.] Prover sends polynomials $\hp_d$ and claims that $\hp_d\equiv p_d$ for each $d \in [D]$. 
Verifier computes $B_d$ inductively for $d\in [D]$ as follows. 

Assume that, for some $d\in [D-1]$, he has the set $B_d$ with $\tdist[v]$ labeled on each vertex $v\in B_d$; this holds initially as he has stored $B_1$. He computes $p_d(r)$ using \cref{eq:dist-poly} and checks whether $\hp_d(r)=p_d(r)$. If the check passes, he believes that $\hp_d \equiv p_d$ and evaluates $\hp_d(u)$ for each $u\in V\setminus B_d$ and constructs $N_{d+1}$ using \cref{eq:turn-exact-nbhd}. Then, $B_{d+1}$ is given by $N_{d+1}\uplus B_d$. 

After $B_D$ is obtained, we get all vertices reachable from $s$ along with their distances from~$s$. Finally, Verifier checks if the other vertices are indeed unreachable from~$s$ by verifying that there is no cross-edge between $B_D$ and $V\setminus B_D$, i.e., if $E\cap (B_D\times (V\setminus B_D)) = \varnothing$. (Intersection scheme, see Fact \ref{fact:set-schemes})
\item[Error probability.] Verifier uses the same element $r$ for $O(D)$ invocations of the sum-check protocol, where each application
of the sum-check protocol is to a univariate polynomial of degree $O(Wn)$. Choosing $|\FF|> DWn^2$, the soundness error for each invocation of the sum-check protocol is at most $1/(Dn)$. Taking a union bound over all $O(D)$ invocations, we get that the total error probability of the protocol is at most $O(1/n)$. 

\item[Help and verification costs] We have $\deg p_d = 
O(Wn)$ for each $d\in [D]$ and hence, hcost is $\tO(DWn)$. Verifier needs to store all vertices and $\tA(v,r)$ for each $v\in [n]$, and hence, vcost is 
$\tO(n)$. The final disjointness can be checked by an $[n,n]$ intersection scheme.
\end{protocol}

\iflipics
\ssspturnstile*
\else
\begin{theorem}\label{thm:sssp-turnstile}
There is a $[DWn,n]$-scheme for \SSSP in the turnstile weight update model.
\end{theorem}
\fi

\mypar{Vanilla Stream\iflipics (Proof of \Cref{thm:sssp-atomic})\else\fi} We now describe a protocol for \SSSP in the model where the edges arrive with their weights, without any further update on them. This is the ``vanilla'' streaming model. 

At the end of the stream, Prover sends the distances $\tdist[v]$ and prev$[v]$--- the parent of $v$ in the shortest path tree rooted at $s$---for all $v\in V$. Verifier checks whether the edges and their weights implied by this proof are correct, using a $[Wn,n]$ subset scheme. Thus, if Prover is honest, we get the distance as well as shortest path from $s$ to each vertex. But we also need to check that there is no path to any vertex shorter than the ones claimed by Prover. We describe a protocol for this. 

For $u,v\in V$ and
$w\in [W]$, define the indicator
function $f$ as $f(u,v,w) =1$ iff $A(u,v)=w$. Let $\tf$ be the
$\FF$-extension of~$f$, for some large
finite field $\FF$.

Retain the definitions of $B_d$ and $N_d$ from last section with the definition of the polynomial $p_d$ changed to 
\begin{equation}\label{eq:alt-dist-poly}
    p_d(U) = \sum_{v\in B_d} \tf(v,U,d+1-\text{dist}_s[v])
\end{equation} 
Hence, it still holds that 
\begin{equation}\label{eq:van-exact-nbhd}
    N_{d+1}= \left\{u \in V\setminus B_d :  p_{d}(u) \neq 0\right\}\,. 
\end{equation}
\begin{protocol}
\item[Stream processing.] The stream updates are of the form $(u,v,w)$ denoting that $A(u,v)=w$. Verifier picks $r\in_R \FF$ and maintains $\tf(v,r,w)$ for each $v\in V$ and $w\in [W]$. He also stores the set $B_1$ with $\text{dist}_s$ labels set to $1$ for each vertex in the set. 
\item[Help message processing and verification.] This part is similar to the turnstile weight update protocol. Of course, this time, the Verifier computes $p_d(r)$ using \Cref{eq:alt-dist-poly}. 
\item[Error probability.] Each polynomial $p_d$ has degree $O(n)$. Verifier does sum-checks for $O(D)$ such polynomials. Choosing $|\FF| \gg Dn$, we can make the error probability small by union bound. 

\item[Help and Verification costs.]
Since the degree of each $p_d$ is at most $n$, the total hcost
is $\tO(Dn)$. Verifier stores $\tf(v,r,w)$ for each $v\in V$ and
$w\in [W]$, which requires $\tO(Wn)$ space. We also need to
store all vertices as we go on assigning the distance labels.
Hence, the total vcost of this protocol is $\tO(Wn)$. 
\end{protocol}

\iflipics
\ssspatomic*
\else
\begin{theorem}\label{thm:sssp-atomic}
There is a $[Dn,Wn]$-scheme for \SSSP in the vanilla streaming model.
\end{theorem}
\fi
\fi

\bibliography{refs}

\iflipics
\else
\bibliographystyle{alpha}
\fi

\end{document}